\documentclass[11pt]{article}

\usepackage{fullpage, amsmath, amssymb, amsthm, bbm, color, dsfont}
\usepackage{graphicx,caption,subcaption,placeins}
\usepackage[ruled,vlined]{algorithm2e}
\usepackage{setspace}

\makeatletter
\def\verbatim@font{\linespread{1}\normalfont\ttfamily}
\makeatother

\usepackage{natbib}
\bibpunct{(}{)}{;}{a}{,}{,}
\usepackage{enumerate}

\newtheorem{theorem}{Theorem}
\newtheorem{proposition}{Proposition}
\newtheorem{theorem*}{Theorem}
\newtheorem{proposition*}{Proposition}
\newtheorem{corollary*}{Corollary}

\newtheorem{lemma}{Lemma}

\newtheorem{remark}{Remark}

\usepackage[bookmarks=true,bookmarksnumbered=true, hypertexnames=false,breaklinks=true]{hyperref}

\title{Empirical Bayes Multistage Testing for Large-Scale Experiments}%
\author{Hui Xu, Weinan Wang}
\begin{document}

\maketitle

\pagenumbering{gobble}

\begin{abstract}
Modern application of A/B tests is challenging due to its large scale in various dimensions, which demands flexibility to deal with multiple testing sequentially. The state-of-the-art practice first reduces the observed data stream to always-valid p-values, and then chooses a cut-off as in conventional multiple testing schemes. Here we propose an alternative method called AMSET (adaptive multistage empirical Bayes test) by incorporating historical data in decision-making to achieve efficiency gains while retaining marginal false discovery rate (mFDR) control that is immune to peeking. We also show that a fully data-driven estimation in AMSET performs robustly to various simulation and real data settings at a large mobile app social network company.

\medskip 

\noindent {\it Keywords:} Always valid inference; Multiple testing; Sequential A/B testing
\end{abstract}


\singlespacing
\clearpage
\pagenumbering{arabic}

\section{Introduction}
Large-scale online experiments \citep{Box05,Ger12} are routinely conducted at various technology companies (Google, Netflix, Microsoft, LinkedIn, etc.) \citep{Xie16,Koh09,kohavi2013online, Tang10,Bak14,Xu15} to help gauge quantifiable impacts of potential feature and product launches. At a large mobile app social network company, hundreds of experiments are conducted on a daily basis, aiming to improve user engagement or app-performance. To accurately measure these changes, hundreds and thousands of metrics are further pumped through our in house A/B platform where routine statistical tests are conducted to judge if the treatment effects are indeed statistically significant. 

Given such high velocity in experiments, we are able to gather a rich repository of historical data related to metric sensitivities and treatment effects, which could potentially be team and product specific. However, traditional $p$-value based decision making that's calculated via two-sample $t$-tests fail to leverage such data, since the test statistic relies only on a single experiment's single metric data. This inefficiency not only translates to loss in statistical power \citep{sun2007oracle} in treatment effect detection, but also limits the type of decision making available to us that is otherwise possible from an empirical Bayesian perspective.

Furthermore, although best practices inform experiment practitioners to hold off until the end to examine experiment efficacy, oftentimes people "hunt" for significance ($p$-value$<0.05$) throughout the whole experiment duration. The incentive for continuously monitoring experimental results is strong since detecting true effects and giving up null effects as early as possible reduces the opportunity cost from running longer experiments, especially when estimating optimal run-time in advance is difficult. However, if no corrections are made, continuous monitoring would augment the type I error incurred, as demonstrated in \citep{Jo15}. Due to invalid $p$-values under continuous monitoring, false discovery rate (FDR) control is also compromised. 

Combining these two pivotal needs, efficient leverage of historical data and continuous monitoring without the pain of type I error inflation, we propose an empirical Bayes procedure named {\textbf{AMSET}} (adaptive multistage empirical Bayes test), together with a fully data-driven estimation procedure that is shown to perform robustly to various simulation and real data settings. Comparing with the state-of-the-art method of \textbf{always valid p-value} utilized by Optimizely \citep{johari2021always}, our method AMSET can control mFDR level at arbitrary stopping time with better precision, and also shows significant power gains.

\subsection{Main contributions}
Our main contribution is two-folds. First, in the oracle situation where population level parameters are known, we derived a Bayes optimal decision rule based on a weighted classification problem and showed that it minimizes the number of missed discoveries subject to marginal false discovery rate (mFDR) control under mild regularity conditions at any time stage. Our second contribution is that we developed an adaptive multistage empirical Bayes (AMSET) procedure by approximating the Bayes optimal decision rule with compound thresholding together with a fully data-driven estimation scheme. We proved that the oracle AMSET gives valid marginal false discovery rate (mFDR) control at any stopping time, thus allowing continuous monitoring of A/B test platforms. In addition, we demonstrated significant power gain of AMSET at mFDR level $\alpha$ as compared to the state-of-the-art always valid p-value approach in various simulation and real data scenarios for both the oracle and data-driven settings.

\subsection{Organization of the paper}
The remainder of the paper is organized as follows. In section 2, we present related literature. In section 3, we describe our setup and notations. In section 4, we establish optimality of the chosen test statistic in fixed horizon setting and introduce our proposed AMSET procedure with known population level parameters, which gives valid mFDR control at any stopping time. In section 5, we describe a fully data-driven procedure that estimates population level parameters using empirical Bayes methods. In section 6, we do a brief recap of the state-of-the-art method of always-valid p-value and clarify connections to our test statistics. In section 7, we compare our method in both the oracle and data-driven scenario with the method used by Optimzely in simulations of fixed-horizon and stagewise testing. In section 8, we demonstrate our method in a real data example using data from large scale A/B testing platform.  

\section{Background and Related Work}
Multiple hypothesis testing can be viewed as a compound decision problem \cite{robbins1951asymptotically} in the sequential context, where each component is testing of a single hypothesis based on sequentially obtained observations. 
\subsection{Sequential hypothesis testing}
Sequential A/B test framework was initially formulated by Wald in 1945 and was widely used in pharmaceutical clinical trials. It is known that the sequential probability ratio test (SPRT) is optimal (\cite{wald1945sequential}) for the single simple vs simple hypothesis testing problem, such that it minimizes the expected number of observations among all sequential procedures with given levels of type I and type II errors. For a comprehensive understanding of related theory, see \cite{berger2013statistical} and \cite{siegmund2013sequential} for details. Adoption of existing sequential tests in online A/B test platforms is limited since validity for arbitrary stopping time is not guaranteed. 

There has been recent interest in constructing valid sequential tests that allow for continuous monitoring. In \cite{Joha17}, \cite{johari2021always}, the authors proposed a general framework of always valid p-values and suggested construction using a particular family of sequential tests called mixture sequential probability ratio tests (mSPRT) \cite{robbins1970statistical}, which are test of power one \cite{robbins1974expected} and is asymptotically optimal if the cost per experiments diminishes to zero \cite{lai2001sequential}. In \cite{wasserman2020universal}, split likelihood-ratio test (split LRT) was introduced and extended to yield always valid p-values and confidence sequences. 

It is worth mentioning that there are other goals besides any time validity that are of interest in sequential testing. For example, constructing sequences of intervals that contain an unknown parameter with uniformly high
probability over an infinite time horizon with Law of the Iterated Logarithm (\cite{darling1967confidence}, \cite{balsubramani2015sequential}), and using sequential A/B testing to find the best alternative with multi-armed bandits (\cite{bubeck2012regret}, \cite{lattimore2020bandit}, \cite{scott2015multi}).

\subsection{Multiple hypothesis testing}
Starting from multiple testing in the single-stage design, the standard procedure to control false discovery rate (FDR) is Benjamini-Hochberg (\cite{benjamini1995controlling}), abbreviated as BH. Sun and Cai proposed an adaptive z-value based procedure \cite{sun2007oracle} that is optimal in the sense that it controls mFDR at level $\alpha$ with the smallest mFNR, where mFNR is the proportion of the expected number of nonnulls among the expected number of nonrejections. 

Due to the feasibility of conducting large-scale online experiments on a routine basis, there has also been recent interest in achieving multiple hypothesis testing controls in the sequential contexts. \cite{wang2017sparse} built on \cite{sun2007oracle} and constructed a simultaneous multistage adaptive ranking and thresholding procedure that mimics the optimality of SPRT, i.e. minimizes the expected total number of samples while satisfying pre-specified constraints on both false positive rate (FPR) and missed discovery rate (MDR). Similarly as SPRT, the procedure is not compatible with continuous monitoring. 

The only state-of-the-art method that provides always valid FDR control in sequential multiple hypothesis testing is to combine BH procedure with always valid p-values (\cite{Joha17}, \cite{johari2021always}). While there are other efforts to obtain always valid multiple testing control in the sequential context, works in this area consider a different data stream from the one we focus on in this paper. For example, a class of methods called ``$\alpha$-spending" and ``$\alpha$-investing" was used to control the family wise error rate (FWER) or the false discovery rate (FDR), where a different hypothesis is evaluated at each time stage (\cite{demets1994interim}, \cite{foster2008alpha}, \cite{javanmard2018online}). There are also works that provide alternative multi-armed bandit (MAB) frameworks for controlling false discoveries when multiple experiments are run over time. In \cite{yang2017framework}, sequence of fixed A/B tests are replaced with best-arm MAB instances. In \cite{jamieson2018bandit}, traffic is adaptively assigned to treatment with the aim of minimizing sample size to reach target true positive proportion, while ensuring anytime false discovery control. Besides, there are also works that focus on combining continuous monitoring with other metrics beyond FDR (\cite{shi2022dynamic}, \cite{malek2017sequential}, \cite{shi2021online}, \cite{gronau2019informed}).

In sequential testing, it may not be necessary to measure all coordinates for all time stage, especially when running long experiments are costly. Adaptive sampling can be used in the sense that the data sampling procedure varies across coordinates. A plethora of powerful multistage testing and estimation have been developed under this more flexible data sampling regime, such as hierarchical testing procedures for pattern recognition (\cite{blanchard2005hierarchical}, \cite{meinshausen2009efficient}, \cite{sun2015hierarchical}), distilled sensing and sequential thresholding methods for sparse detection (\cite{haupt2011distilled}), \cite{malloy2011sequential}, \cite{malloy2014sequential}), etc. $\alpha-$investing can be considered as an extreme case of adaptive sampling where each hypothesis is considered once. In this paper, we will focus on a particular type of adaptive sampling where the set of hypothesis to sample from decreases or filters down at each time stage. 

\section{Setup and Notations}
\label{setup}
Consider the well established problem (\citep{sun2007oracle,jin2007estimating, wang2017sparse,tony2019covariate}) of recovering the support of a high dimensional vector  $\boldsymbol{\theta}= (\theta_1,\ldots, \theta_m) \in \{0,1\}^m$ from i.i.d. sequential measurements of variables $X_1,\ldots, X_m$. At each time stage $t = 1,2,\ldots$, we observe an independent measurement of each variable,
\begin{equation}
\label{model}
   X_{it} \sim (1-\theta_i)F_0 + \theta_i F_{i}, 
\end{equation}
for $i = 1, \ldots, m$. That is to say, $X_{it}$ follows null distribution $F_0$ if $\theta_i = 0$ and the non-null distribution $F_{i}$ if $\theta_i = 1$. In the context of large-scale A/B testing, $F_0$ is the distribution of two-sample t-statistics with mean zero approximated by standard normal and $F_i$'s are approximately normal distributions with different means corresponding to effect sizes.

\begin{figure*}[!h]
    \centering
    \includegraphics[width=\linewidth]{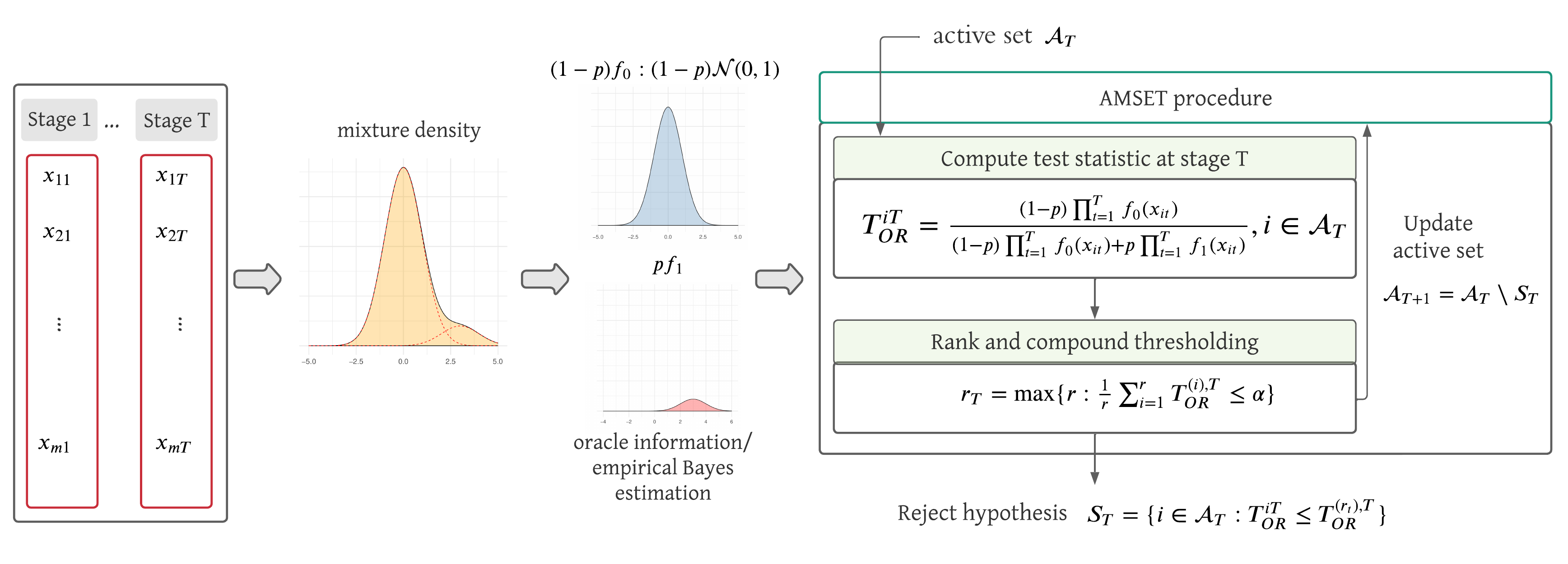}
    \caption{Illustration of AMSET procedure. The left part of the diagram corresponds to multistage data and mixture deconvolution obtained either by oracle or data-driven estimations. In real settings, alternative density $f_1$ could be Gaussian mixtures.}
    \label{illustration}
\end{figure*}

Our goal is to develop a cost effective and (asymptotically) valid decision rule to identify the support of $\boldsymbol{\theta}$ (i.e. non-zero treatment effects) reliably in the sequential testing setting. To be more precise, let $\boldsymbol{\delta}^t = (\delta_1^t,\ldots, \delta_m^t)$ denote the decision rule based on $\mathcal{F}^t = \sigma( \{x_{is}\}_{i \in [m], s \leq t})$, where $\delta_i^t \in \{0,1\}$ indicates whether coordinate $i$ is classified as a null or non-null case. For any fixed time stage $t$, define \textbf{False Discovery Rate (FDR)} as $\mbox{FDR}(\boldsymbol{\delta}^t) = \mathbb{E}[\mbox{FDP}(\boldsymbol{\delta}^t)]$, where 
\begin{equation}
\label{fdr}
    \mbox{FDP}(\boldsymbol{\delta}^t) = \frac{\sum_{i=1}^m (1-\theta_i)\delta_i^t}{\left(\sum_{i=1}^m \delta_i^t \right)\vee 1}.
\end{equation}
An asymptotically equivalent measure is the \textbf{marginal false discovery rate (mFDR)} (\cite{genovese2002operating}) defined as 
\begin{equation}
\label{mfdr}
    \mbox{mFDR}(\boldsymbol{\delta}^t) = \frac{\mathbb{E}[\sum_{i=1}^m (1-\theta_i) \delta_i^t]}{\mathbb{E}\left[\left(\sum_{i=1}^m \delta_i^t \right)\vee 1\right]}.
\end{equation}

While mFDR or FDR guarantees validity, cost effectiveness can be measured by \textbf{missed discovery rate (MDR)}, 

\begin{equation}
\label{mdr}
        \mbox{MDR}(\boldsymbol{\delta}^t) = \mathbb{E}\left[ \frac{\sum_{i=1}^m (\theta_i (1-\delta_i^t)}{\left(\sum_{i=1}^m \theta_i \right)\vee 1}\right].
\end{equation}
The goal of large-scale A/B testing can therefore be formulated as minimizing MDR subject to FDR (or mFDR) control. For any stopping time $\tau$, we can similarly define $\mbox{FDP}(\boldsymbol{\delta}^\tau)$, $\mbox{mFDR}(\boldsymbol{\delta}^\tau)$, and $\mbox{mFDR}(\boldsymbol{\delta}^\tau)$, where the expectation is with respect to all randomness including stopping time $\tau$. 

\section{Sequential A/B Test}

We first consider an ideal situation where an oracle knows all population level parameters including the proportion of non-nulls $p$, as well as the densities of null $f_0$ and non-null $f_1$. 

\subsection{Optimal sequential test for fixed horizon}
\label{sec-optimality}
We start by restricting to a non-adaptive setting, in the sense that data for each coordinate is collected at all time stages. Suppose that $\theta_i$ are iid Bernoulli random variables with $\mathbb{P}(\theta_i = 1) = p$. Define the following oracle test statistic for stage $T$ and coordinate $i$,
\begin{equation}
\label{eq-Tor}
T_{OR}^{iT} = \frac{(1-p)\prod_{t=1}^Tf_0(x_{it})}{(1-p)\prod_{t=1}^T f_0(x_{it})+p\prod_{t=1}^Tf_1(x_{it})}.
\end{equation}
Since we assume that the $x_{it}$ are independent draws from the same distribution for all time stages, the oracle test statistic is the local false discovery rate (lfdr),  $\mathbb{P}(\theta_i = 0 | x_{i,1:T})$ (\cite{efron2001empirical}), given the data stream up to time stage $T$. Truncation with lfdr was found to be the optimal test in many situations of both sequential testing (\cite{berger2013statistical}) and multiple testing (\cite{sun2007oracle}). 

For any time stage $T$ and mFDR level of $\alpha \in (0,1)$, we define a weighted classification problem with the following loss function
\begin{equation}
\label{loss}
    L_\lambda(\theta, \delta^T) = \frac{1}{m}\left[\sum_{i=1}^m \lambda (1-\theta_i) \delta_i^T + \theta_i (1-\delta_i^T) - \alpha \lambda \delta_i^T\right],
\end{equation}
where $\lambda >0$ is the relative weight for false positive. The following theorem shows that $T_{OR}^{iT}$ minimizes the classification risk $\mathbb{E}[L_{\lambda}(\theta, \delta^T)]$ at any time stage $T$.  

\begin{theorem}
\label{classification}
Suppose that $\theta_i$, $i = 1,\ldots, m$, are iid Bernoulli$(p)$ random variables, and $x_{it}$ are independent draws from the mixture model in equation \eqref{model} for all $t= 1,\ldots, T$. Given known values of $p, f_0, f_1$, the classification risk $\mathbb{E}[L_{\lambda}(\theta, \delta^T)]$ is minimized by $\delta^B(\lambda,T) = (\delta_1^B(\lambda,T), \ldots, \delta_m^B(\lambda,T))$ for any $T \geq 1$, where 
\begin{equation}
\label{optimal}
    \delta_i^B(\lambda, T) = 1\Big\{T_{OR}^{iT} \leq \frac{1+\alpha\lambda}{1+\lambda}\Big\}.
\end{equation}
\end{theorem}

The above theorem shows that the Bayes optimal sequential test for the weighted classification problem is a truncation rule based on lfdr statistic. We then make connections with the optimal solution in the large-scale A/B testing problem. For the sake of notation, let $G_0^{iT}$ and $G_1^{iT}$ be the cumulative distribution functions (cdf) of $T_{OR}^{iT}$ under the null and the alternative, respectively, then the marginal distribution for $T_{OR}^{iT}$ can be written as $(1-p)G_0^{iT} + pG_1^{iT}$.
\begin{theorem}
\label{thm-optimality}
Under the same data generating process as in Theorem \ref{classification}. Let $1-p > \alpha$. Assume that $G_0^{iT}$ and $G_1^{iT}$ are continuous functions and 
\begin{equation}
\label{optimality-condition}
    \frac{(1-p)G_0^{iT}(\alpha)}{(1-p) G_0^{iT}(\alpha)+ pG_1^{iT}(\alpha)} < \alpha,
\end{equation}
then there exists a $\lambda^*>0$ such that the solution to the weighted classification problem in equation \eqref{optimal} is also optimal in the multiple testing problem at time stage $T$ in the sense that, among all decision rules at time $T$ that control mFDR at level $\alpha$, it has the smallest expected proportion of missed discoveries. That is, $\delta^{B}(\lambda^*,T)$ is the solution to the constraint optimization problem after observing data until stage $T$.
\begin{align}
\begin{split}
\label{optimization}
    \min_{\substack{(\delta_i)_{i=1}^m: \\ \{x_{it}\}_{i \in [m], t \leq T} 
    \mapsto \{0,1\}^m}} & \quad   \frac{1}{m}\sum_{i=1}^m \mathbb{E}\left[\theta_i (1-\delta_i)\right]\\
    \textrm{s.t.} & \quad \text{mFDR}(\delta)= \frac{\mathbb{E}\left[\sum_{i=1}^m(1-\theta_i)\delta_i\right]}{\mathbb{E}\left[\left(\sum_{i=1}^m \delta_i \right)\vee 1\right] } \leq \alpha.
\end{split}
\end{align}
\end{theorem}

The above theorem establishes that truncation of lfdr is also optimal in the multiple testing problem at any fixed time stage $T$. For large-scale A/B testing, the signal is usually sparse, where the proportion of non-nulls $p$ is usually in the range of $(0.01, 0.1)$, thus satisfying the weak assumption $1-p > \alpha$. 

The assumption of \eqref{optimality-condition} is equivalent to the statement that $\text{mFDR}$ of simple trucation rule at level $\alpha$ is less than $\alpha$.

\begin{proposition}
\label{prop-optimality-condition}
Under the same data generating process as in Theorem \ref{classification}. For any coordinate $i$ and fixed time stage $T$, 
\begin{equation}
\label{eq-condition}
    \frac{(1-p)G_0^{iT}(\alpha)}{(1-p)G_0^{iT}(\alpha) + p G_1^{iT}(\alpha)} \leq \alpha.
\end{equation}
\end{proposition}

From the above proposition, we know that assumption \eqref{optimality-condition} only requires the inequality to be strict. That is, the simple thresholding rule is more conservative than the nominal level $\alpha$. In other words, there is additional gain by increasing the truncation threshold without overshooting in terms of $\text{mFDR}$. 

We have motivated the choice of oracle test statistic by establishing optimality of lfdr-based test in fixed time (horizon) setting. Given observed data until fixed time stage $T$, the decision rule  $\delta^B(\lambda^*, T)$ minimizes missed discovery rate among all tests that control mFDR at level $\alpha$. Since power is $1-MDR$, it is equivalent to maximizing power subject to mFDR control. 

\subsection{Adaptive sequential test}
\label{sec-adaptive}
Instead of sampling for every coordinate at each time stage, we may want to consider adaptive sampling, in order to achieve some balance between power and sampling cost. While we have shown in the previous section that the Bayes optimal test at every time stage is a thresholding function based on observations at all past time stages, it might suffer from sample inefficiency, which is undesirable for monetization A/B tests. For this purpose we consider the following adaptive procedure

\begin{align*}
    &\text{Stop sampling and reject for coordinate $i$ at time stage } \\ 
    &N_i = \min\Big\{t \geq 1: T_{OR}^{it} \leq \frac{1+\alpha\lambda_t}{1+\lambda_t}\Big\},
\end{align*}
where $\lambda_t$ could vary across time points. 
In other words, we maintain an active set $\mathcal{A}_t = \{i: N_i \geq t\}$ for sampling and testing at each time stage $t$. 

The adaptive sampling procedure allows us to save cost from data collection and further computation of test statistic. Adaptive sampling will not compromise global mFDR control, and the thresholding rule based on our oracle test statistic is stage-wise optimal (conditioning on stagewise active set) as shown in Lemma \ref{lem-optimality}.

\subsection{Oracle AMSET}

In this section, we will introduce the method of AMSET (adaptive multistage empirical Bayes test) that combines the oracle test statistics motivated in \ref{sec-optimality} and the adaptive procedure described in \ref{sec-adaptive}, to arrive at a sequential procedure for multiple testing that has valid mFDR control at any stopping time, and thus is immune to peeking.  

Let $\mathcal{A}_t$ be the active set at stage $t$. Starting with the full active set $\mathcal{A}_1 = \{1,\ldots, m\}$ for all $i \in \{1,\ldots, m\}$. We truncate the lfdr test statistics using compound thresholding (\cite{sun2007oracle}) and update the active set at every time stage $t$ for a pre-specified level $\alpha > 0$ with details listed in Algorithm \ref{algo-oracle}, and a diagram illustration in Figure \ref{illustration}. The non-adaptive version that we call multistage empirical Bayes test (MSET) in Algorithm \ref{algo-mset} makes decisions based on full observations at each coordinate and all time stages and therefore has higher power. In general we recommend AMSET over MSET for monetization A/B tests in order to save experimentation cost as a tradeoff for power.  

\RestyleAlgo{ruled}
\begin{algorithm}[!h]
\caption{Oracle AMSET}
\label{algo-oracle}
Set $\alpha$ and start with full active set $\mathcal{A}_1 = \{1,\ldots, m\}$\;
Specify any stopping criterion\;
\For{time stages $t= 1, 2,\ldots$}{
(\textbf{Data Sampling}) Observe data $\{x_{it}\}_{i \in \mathcal{A}_t}$\;
(\textbf{Ranking}) For all $i \in \mathcal{A}_t$, compute $T_{OR}^{it}$ in Eq \eqref{eq-Tor} based on $\{x_{is}\}_{s \leq t}$ and sort them in ascending order $$T_{OR}^{(1),t}, T_{OR}^{(2),t}, \ldots, T_{OR}^{(k_t),t},$$ where $k_t = \mbox{Card}(\mathcal{A}_t)$\;
(\textbf{Compound Thresholding}) Define threshold 
\begin{equation*}
r_t = \max\Big\{r: \frac{1}{r} \sum_{i= 1}^{r} T_{OR}^{(i),t} \leq \alpha\Big\}
\end{equation*}
 Let $S_{t} = \{i \in \mathcal{A}_t: T_{OR}^{it} \leq T_{OR}^{(r_t), t}\}$ \;
(\textbf{Updating}) Let $\mathcal{A}_{t+1} = \mathcal{A}_t \setminus S_t$\;
(\textbf{Decision}) Let $\delta_i^t = 1$ for all $i \notin \mathcal{A}_{t+1}$ and $0$ otherwise\;
Break if stopping criterion is satisfied or $\mathcal{A}_{t+1} = \emptyset$.}
\end{algorithm}

The procedure AMSET offers flexibility in terms of cost saving since there is no need to collect data once a specific coordinate is rejected. That is to say, a rejected hypothesis will always remain rejected. AMSET also allows us the convenience to continuously monitor the sequential testing results with mFDR guarantee at any stopping time as stated in the following Theorem. 

\begin{theorem}
\label{thm-oracle}
For any stopping time $\tau$, let $\delta^\tau = (\delta_1^\tau,\ldots, \delta_m^\tau)$ generated by Algorithm \ref{algo-oracle}. Then
\begin{equation*}
    \textup{mFDR}(\mathbf{\delta}^\tau) \leq \alpha.
\end{equation*}
We also have fixed-horizon mFDR and FDR control at level $\alpha$, i.e.
$$
\text{FDR}(\delta^t | S_t) = \mathbb{E}\left[\frac{\sum_{i \in S_t} (1-\theta_i) \delta_i^t}{\left(\sum_{i \in S_t} \delta_i^t\right) \vee 1}\right] \leq \alpha,
$$
and
$$
\textup{mFDR}(\delta^t |S_t) = \frac{\mathbb{E}\left[\sum_{i \in S_t} (1-\theta_i) \delta_i^t\right]}{\mathbb{E}\left[\left(\sum_{i \in S_t} \delta_i^t\right) \vee 1\right]} \leq \alpha,
$$
for any fixed $t$.

\end{theorem}

\begin{remark}
The compound thresholding step in the procedure above can be replaced by a simple thresholding, that is, we reject all $i \in \mathcal{A}_t$ such that $T_{OR}^{it} \leq \alpha$. Since all coordinates rejected by simple thresholding will also be rejected by compound thresholding, the former is a more conservative procedure and therefore still satisfies mFDR control at any stopping time. 
\end{remark}

\section{Data-driven Implementation}
\label{sec-dd}
In practice, we do not have the oracle knowledge of true population parameters ($f_1$, $f_0$, and $p$) and need a completely data-driven procedure that can be used for industry applications. Given large scale testing results in historical data, we can harness the power of established empirical Bayes methods in order to estimate the oracle test statistic of $T_{OR}^{it}$. For the purpose of A/B testing, the null distribution is approximately $\mathcal{N}(0,1)$ due to central limit theorem and the alternative distribution is approximately a mixture of gaussians. The difficulty with of data driven procedure thus lies in the estimation of non-null proportion $p$ and deconvolution of gaussian mixtures using relevant historical data. 

\begin{enumerate}
    \item (Estimate $\hat{p}$) For estimation of the non-null proportion $p$, we use the method proposed in \cite{jin2007estimating} based on the empirical characteristic function and Fourier analysis, which is shown to be uniformly consistent over a wide class of parameters.
    \item (Estimate $\hat{f}_1$) For estimation of non-null density, we will use Kiefer-Wolfowitz maximum likelihood estimation for Gaussian mixtures with software ``REBayes" develped by \cite{koenker2017rebayes} and computation supported by \cite{koenker2014convex}. The method belongs to the general estimation strategy of g-estimation and has consistency result as stated in \cite{kiefer1956consistency}.
\end{enumerate}

Additional details of deconvolution and recovery of $\hat{f}_1$ can be found in the Appendix \ref{appen-deconvolve}. Given estimates of non-null proportion $\hat{p}$ and non-null density $\hat{f}_1$, we get the following data-driven test statistic for coordinate $i$ at time stage $T$. 
\begin{equation}
\label{eq-Tdd}
T_{DD}^{iT} = \frac{(1-\hat{p})\prod_{t=1}^T f_0(x_{it})}{(1-\hat{p})\prod_{t=1}^T f_0(x_{it})+\hat{p}\prod_{t=1}^T\hat{f}_1(x_{it})},
\end{equation}
which can then replace the oracle test statistics $T_{OR}^{iT}$ in AMSET (Algorithm \ref{algo-oracle}). We show details of the slightly modified algorithm of data-driven AMSET in Algorithm \ref{algo-dd}. 

\RestyleAlgo{ruled}
\begin{algorithm}[!h]
\caption{AMSET}
\label{algo-dd}
Set $\alpha$ and start with full active set $\mathcal{A}_1 = \{1,\ldots, m\}$\;
Specify any stopping criterion\;
\For{time stages $t= 1, 2,\ldots$}{
(\textbf{Data Sampling}) Observe data $\{x_{it}\}_{i \in \mathcal{A}_t}$\;
(\textbf{Ranking}) For all $i \in \mathcal{A}_t$, compute $T_{DD}^{it}$ in Eq \eqref{eq-Tdd} based on $\{x_{is}\}_{s \leq t}$ and sort them in ascending order $$T_{DD}^{(1),t}, T_{DD}^{(2),t}, \ldots, T_{DD}^{(k_t),t},$$ where $k_t = \mbox{Card}(\mathcal{A}_t)$\;
(\textbf{Compound Thresholding}) Define threshold 
\begin{equation*}
r_t = \max\Big\{r: \frac{1}{r} \sum_{i= 1}^{r} T_{DD}^{(i),t} \leq \alpha\Big\}
\end{equation*}
 Let $S_{t} = \{i \in \mathcal{A}_t: T_{DD}^{it} \leq T_{DD}^{(r_t), t}\}$ \;
(\textbf{Updating}) Let $\mathcal{A}_{t+1} = \mathcal{A}_t \setminus S_t$\;
(\textbf{Decision}) Let $\delta_i^t = 1$ for all $i \notin \mathcal{A}_{t+1}$ and $0$ otherwise\;
Break if stopping criterion is satisfied or $\mathcal{A}_{t+1} = \emptyset$.}
\end{algorithm}

\section{Comparison with always valid p-value approach}
\subsection{FDR control with always valid p-value}
We compare AMSET (or MSET) with the state-of-the-art method that combines always valid p-values with Benjamini-Hochberg (BH) procedure \cite{benjamini1995controlling}. Recall always valid p-value in \cite{Jo15} is constructed using mSPRT likelihood ratio test statistic (originally proposed in \cite{robbins1974expected}) at stage $T$ for coordinate $i$,
$$
\Lambda_T^H = \int_\Theta \prod_{t=1}^T \frac{f_{\theta}(x_{it})}{f_{\theta_0}(x_{it})} dH(\theta),
$$
where $H(\theta)$ is some prior distribution on parameter $\theta$. 

Always valid p-value process for coordinate $i$ is constructed as 
$$
p_0 = 1, p_T = \min\{p_{T-1}, 1/\Lambda_T^H\},
$$
such that for any stopping time $\tau$,
$$
\forall s \in [0,1], \mathbb{P}_{\theta_0}(p_\tau \leq s) \leq s.
$$


For fair comparison, we use a prior $H$ that has probability $1$ on the true mean of alternative distribution to compare with our oracle procedure, and a prior $H \sim \mathcal{N}(0,\tau^2)$ where $\tau^2 = 2$ to compare with our data driven procedure. 

\subsection{Connection of lfdr to always valid p-values}

Both the oracle and data-driven test statistics of lfdr can also be formulated in the framework of always valid p-values using a standard martingale argument for running MLE LRT in universal inference (\cite{wasserman2020universal}). For the sake of brevity, we will only illustrate construction of always valid p-value using data-driven statistic $T_{DD}^{iT}$ for each coordinate $i$. 

Consider the sequential test that rejects the null and stop at time $T$ if 
\begin{equation}
\label{universal}
L_{iT} := \frac{1-\hat{p}}{\hat{p}} \frac{1-T_{DD}^{iT}}{T_{DD}^{iT}} \geq 1/\alpha.
\end{equation}
Since $L_{iT} = \frac{\prod_{t=1}^T \hat{f}_1(x_{it})}{\prod_{t=1}^T f_0(x_{it})}$ is a nonnegative martingale with respect to the natural filtration, we can use Doob's inequality to bound the Type-I error as  
$$
\mathbb{P}_{H_0} (\exists t \geq 1: L_{it} \geq 1/\alpha) \leq \alpha \mathbb{E}_{H_0}[L_{it}] = \alpha. 
$$
This shows that the sequential test defined by Eq \eqref{universal} is a valid sequential test at level $\alpha$. Since there is a correspondence between sequential tests and always valid p-values, we know that $p_{it} = 1/L_{it}$ is an always valid p-value. That is, for any stopping time $\tau$, $\mathbb{P}_{H_0}(p_{i\tau} \leq \alpha) \leq \alpha$. This means that we can also adopt Optimizely's framework by combining BH procedure with the always valid p-value process constructed with Eq \eqref{universal} to guarantee finite sample FDR control at any stopping time and thus is immune to peeking.

\section{Simulation}
Throughout our simulations, we compare our proposed method of AMSET with the state-of-the-art method used by Optimizely \cite{johari2021always} that combines standard BH procedure with always valid p-value constructed using mixture sequential probability ratio test (mSPRT). Since AMSET is an adaptive procedure that does not sample from all coordinates at all time stages, there will be possible power loss as compared to the non-adaptive MSET or Optimizely's always valid p-value approach.

\begin{figure*}[!ht]
    \centering
    \includegraphics[width=\linewidth]{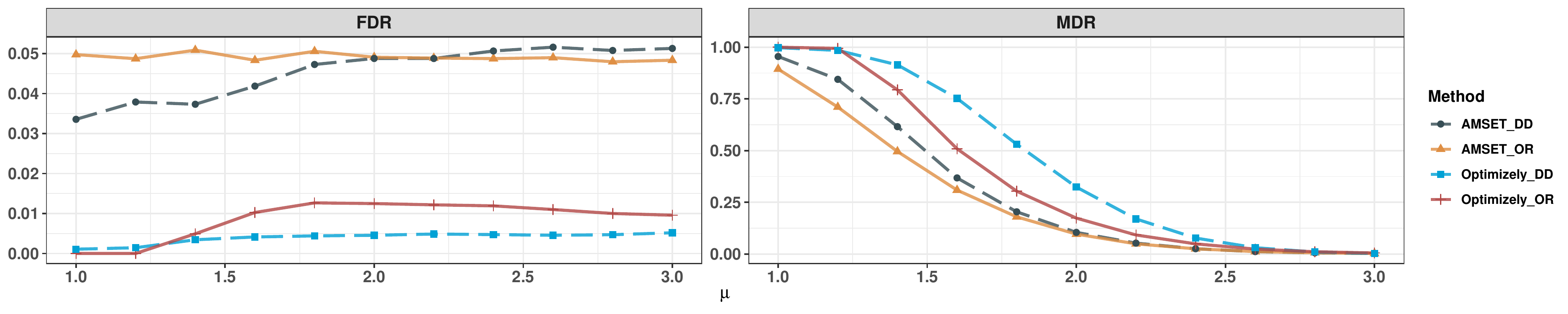}
    \caption{Setting 1 result for fixed horizon comparison. Fix $p = 0.05$, consider $f_1$ as the density of $\mathcal{N}(\mu,1)$, where $\mu$ varies from $1$ to $3$ with grid size of $0.2$. Rates are computed as average over $500$ simulations.} 
    \label{fixed_1}
\end{figure*}

\begin{figure*}[!ht]
    \centering
    \includegraphics[width=\linewidth]{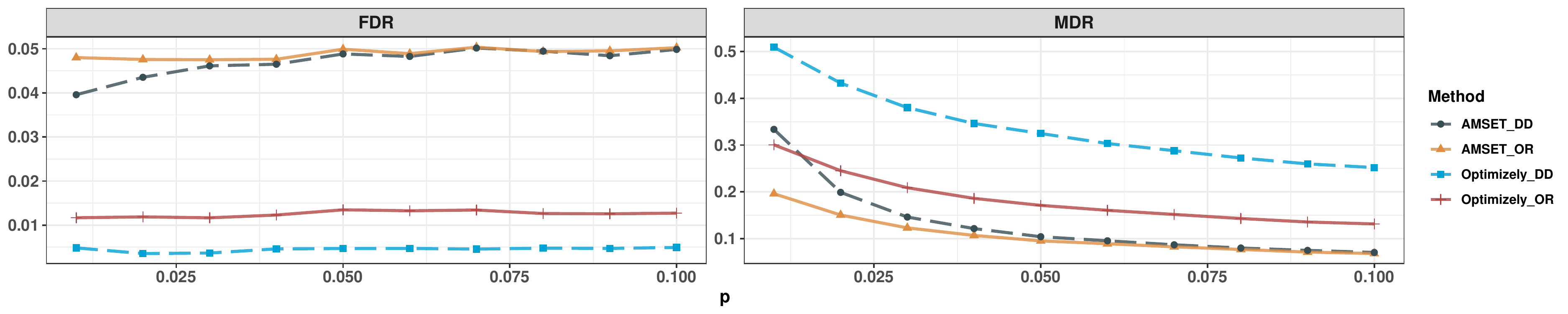}
    \caption{Setting 2 result for fixed horizon comparison. Fix $f_1$ as the density of $\mathcal{N}(2,1)$, vary $p$ from $0.01$ to $0.1$ with grid size $0.01$. Rates are computed as average over $500$ simulations.}
    \label{fixed_2}
\end{figure*}

\begin{figure*}[!ht]
    \centering
    \includegraphics[width=\linewidth]{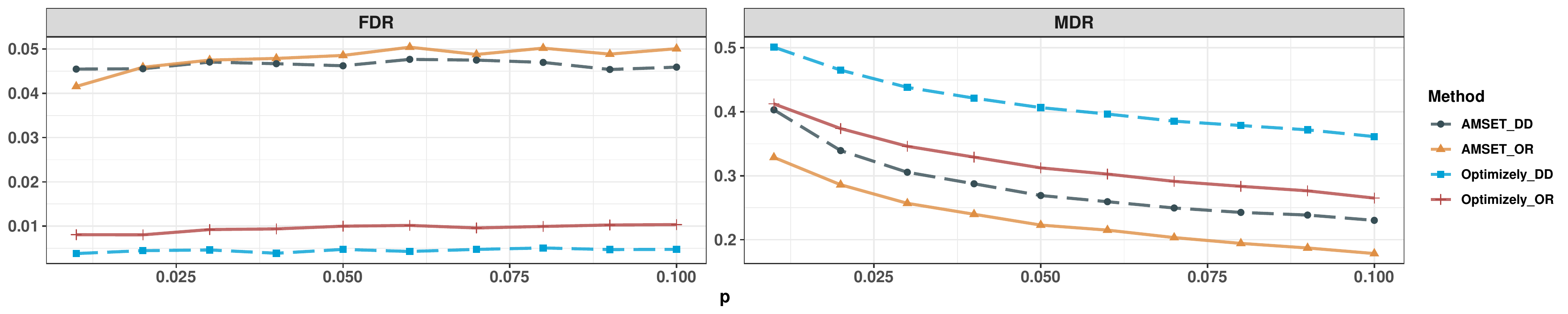}
    \caption{Setting 3 result for fixed horizon comparison. Vary $p$ from $0.02$ to $0.1$ with grid size of $0.01$ and consider $f_1$ as the density of $\mathcal{N}(\mu,1)$, where $\mu \sim \mbox{Unif}(2,4)$ at each location. Rates are computed as average over $500$ simulations.}
    \label{fixed_3}
\end{figure*}

\begin{figure*}[!ht]
    \centering
    \includegraphics[width=\linewidth]{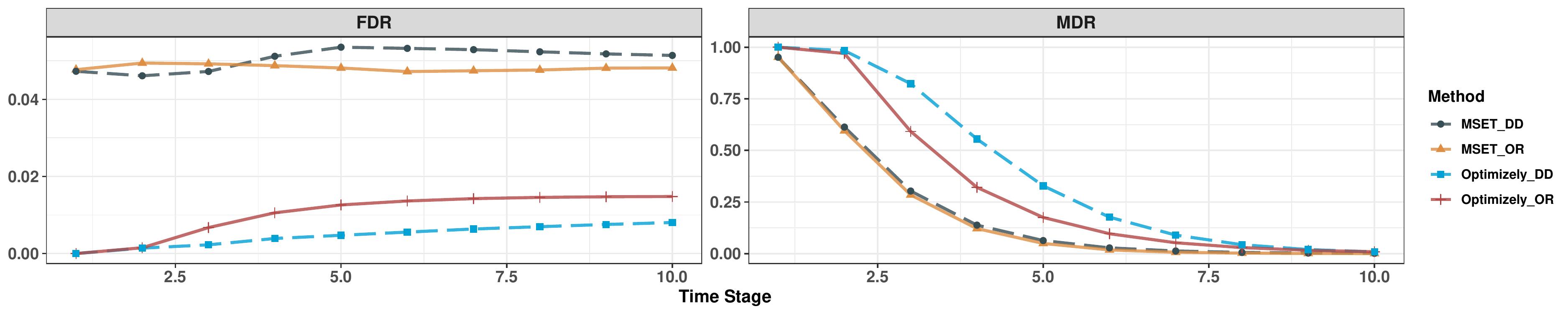}
    \caption{Setting 1 result for stagewise comparison. Fixed $\mu = 2$ for each non-null location and proportion of non-nulls $p = 0.05$. Rates are computed as average over $500$ simulations.}
    \label{stagewise-1}
\end{figure*}

\begin{figure*}[!ht]
    \centering
    \includegraphics[width=\linewidth]{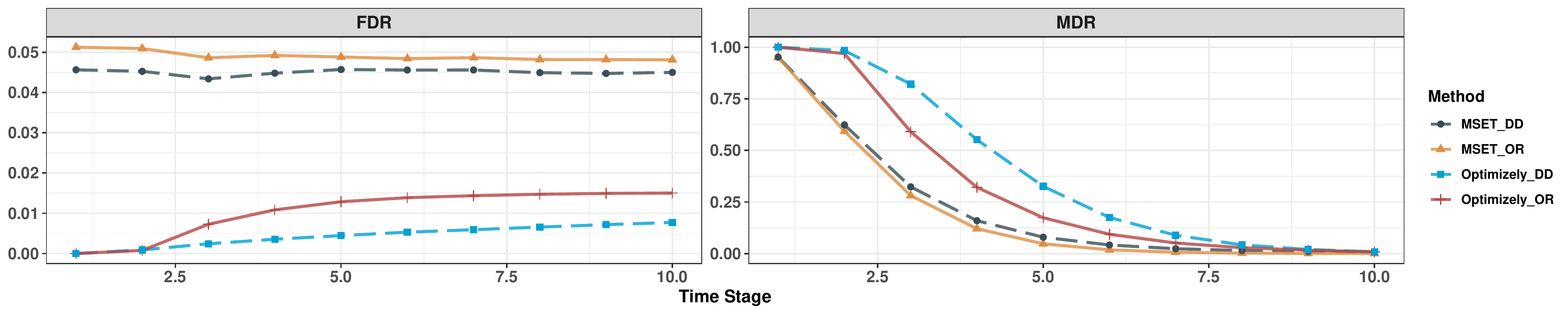}
    \caption{Setting 2 result for stagewise comparison. Random $\mu \stackrel{i.i.d}{\sim} \mbox{Unif}(1,3)$ for each non-null location and proportion of non-nulls $p = 0.05$. Rates are computed as average over $500$ simulations.}
    \label{stagewise-2}
\end{figure*}

\begin{figure*}[!ht]
    \centering
    \includegraphics[width=\linewidth]{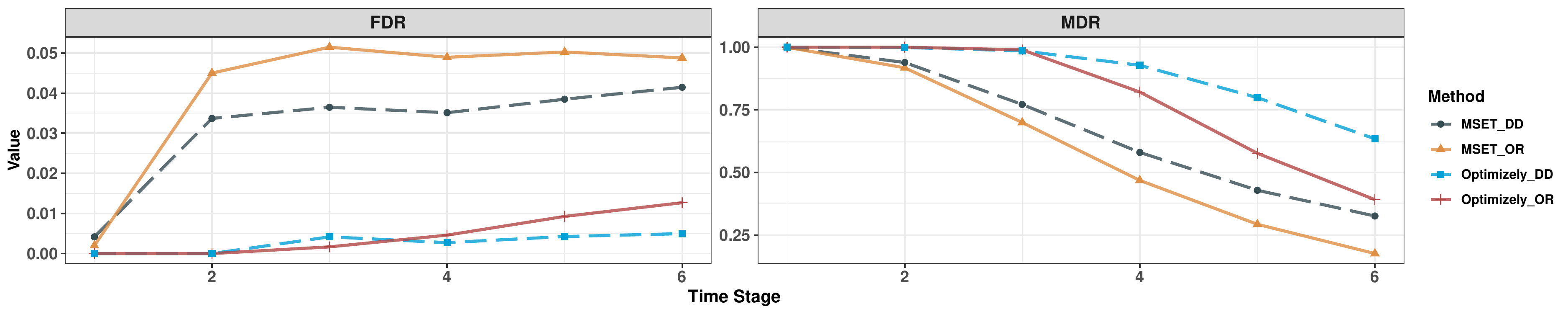}
    \caption{Real data example 1 results. Rates are computed as average over $500$ simulations.}
    \label{snap-setting1}
\end{figure*}

\begin{figure*}[!ht]
    \centering
    \includegraphics[width=\linewidth]{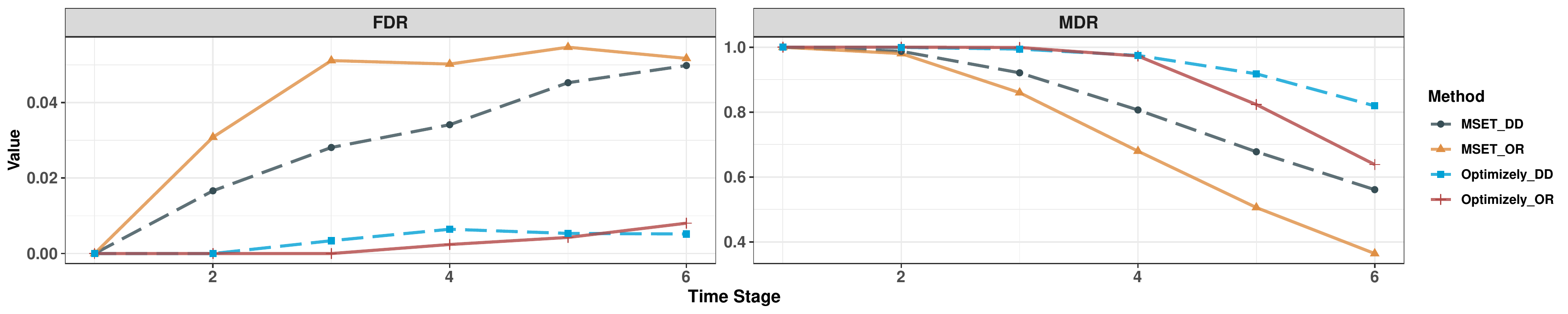}
    \caption{Real data example 2 results. Rates are computed as average over $500$ simulations.}
    \label{snap-setting2}
\end{figure*}

\begin{figure*}[!ht]
    \centering
    \includegraphics[width=\linewidth]{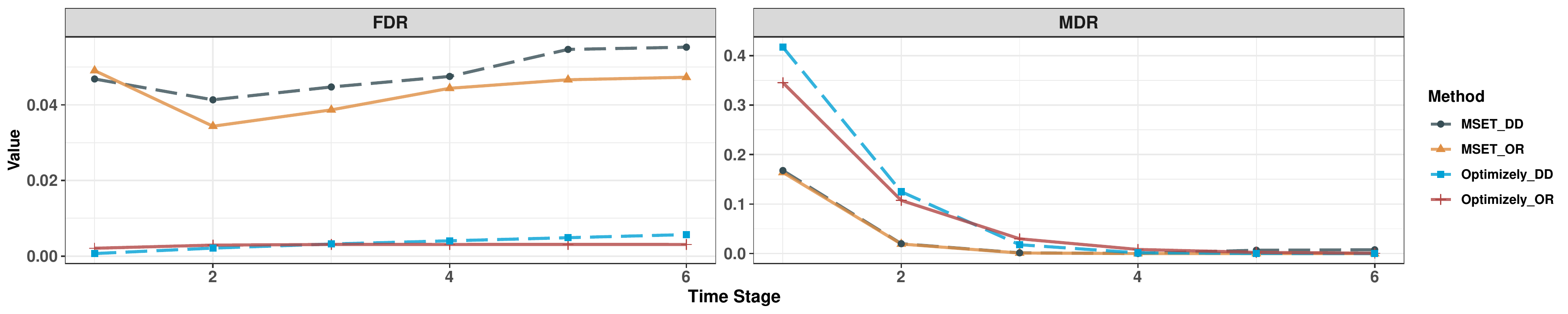}
    \caption{Real data example 3 results. Rates are computed as average over $500$ simulations.}
    \label{snap-setting3}
\end{figure*}

For data generation, we consider $m= 5000$ locations with ground truth $\theta_i$ ($1\leq i\leq m$) drawn i.i.d from Bernoulli distribution with parameter $p$, for some non-null proportion $p \in (0,1)$. In realistic settings, the signal proportion $p$ is usually small ($\leq 0.1$). 

For each time stage, if coordinate $i$ is still in the observation set, we draw an independent sample from $f_0$ if $\theta_i = 0$ and $f_1$ otherwise. In the context of A/B testing, different coordinates corresponds to different randomized experiments and observation $x_{it}$ corresponds to the two-sample t-statistic at time stage $t$ for experiment $i$. Since $f_0$ corresponds to the density of two-sample t-statistic under the null where treatment effect is zero, it can be approximated by standard normal density. Similarly, $f_1$ can be approximated by a Gaussian density with non-zero location $\mu \neq 0$ or a Gaussian mixture with various non-zero locations. 

We consider four different tests labeled in the following way. (1) \verb+AMSET_OR+/\verb+MSET_OR+: AMSET/MSET with oracle statistic (Eq \eqref{eq-Tor});
(2) \verb+AMSET_DD+/\verb+MSET_DD+: AMSET/MSET with data-driven statistic (Eq \eqref{eq-Tdd});
(3) \verb+Optimizely_OR+: always valid p-value approach with oracle prior on true alternative;
(4) \verb+Optimizely_DD+: always valid p-value approach with Gaussian prior $\mathcal{N}(0,2)$. 
For each setting with a fixed set of hyperparameters $p$ and $f_1$, we will fix $\mbox{mFDR}$ (Eq. \eqref{mfdr}) control at nominal level $0.05$ and compare FDR (Eq. \eqref{fdr}), as well as $\mbox{MDR}$ (Eq. \eqref{mdr}) across $500$ simulations. Since $\mbox{MDR} = 1- \mbox{power}$, smaller MDR corresponds to more powerful tests. 

We consider simulations in both fixed horizon and stage-wise comparisons regimes. For the fixed horizon setting, we run the simulated sequential tests until a fixed time stage and compare FDR/MDR metrics at the last time stage. This allows us to observe trends for a wide range of hyper-parameter choices. On the other hand, stage-wise comparison illustrates how the trend evolves over time for a chosen set of hyper-parameters. 

\subsection{Fixed horizon comparison}
For fixed horizon comparison, we set the number of time stages to be fixed at $5$ with the following three scenarios described below.

\begin{enumerate}
    \item Fix $p = 0.05$, consider $f_1$ as the density of $\mathcal{N}(\mu,1)$, where $\mu$ varies from $1$ to $3$ with grid size of $0.2$
    \item Fix $f_1$ as the density of $\mathcal{N}(2,1)$, vary $p$ from $0.01$ to $0.1$ with grid size $0.01$
    \item Vary $p$ from $0.02$ to $0.1$ with grid size of $0.01$ and consider $f_1$ as the density of $\mathcal{N}(\mu,1)$, where $\mu \sim \mbox{Unif}(2,4)$ at each coordinate
\end{enumerate}

From the results presented in Figure \ref{fixed_1}-\ref{fixed_3}, it can be observed that both the oracle and data driven versions of the state-of-the-art method used by Optimizely are much more conservative in term of $\mbox{FDR}$ control and has generally lower power than AMSET across all settings and hyperparameter choices. The difference in power diminishes as the signal strength $\mu$ increases and power goes to $1$, but remains approximately the same for different choices of non-null proportions $p$. 

\subsection{Stagewise comparison}
We also use simulation to compare the oracle and data-driven methods of MSET and Optimizely by looking at FDR and MDR at each time stage for a fixed set of  hyperparameters as follows.
\begin{enumerate}
    \item Fixed $\mu = 2$ for each non-null location and proportion of non-nulls $p = 0.05$
    \item Random $\mu \stackrel{i.i.d}{\sim} \mbox{Unif}(1,3)$ for each non-null location and proportion of non-nulls $p = 0.05$. 
\end{enumerate}

From results presented in Figure \ref{stagewise-1}-\ref{stagewise-2}, it can be similarly observed that always valid p-value approach used by Optimizely is much more conservative in term of $\mbox{FDR}$ control and has generally lower power than MSET procedure. The difference gets smaller as the number of time stages increases and power for all tests goes to $1$. 

\section{Real Data Example}

Instead of simulating with arbitrary hyper-parameters, we apply AMSET on actual A/B testing data. That is, we first extract realistic hyper-parameters from historical data of a pool of experiments via Kiefer-Wolfowitz MLE estimation for Gaussian mixtures. And then compare the stagewise performance of various methods under the learnt hyper-parameters.

In particular, we work with the following three metric settings. 

\begin{enumerate}
\item \verb|DIRECT_SEND____total_sent|
\item \verb|CHAT_CHAT_SEND____chat_video_send|
\item  \verb|TIER1_ACTIVE_DAYS____total_active_days|
\end{enumerate}

For each metric, we can aggregate a set of randomized experiments with potential impact on the given metric and compute standardized $z$-scores as estimators for the average treatment effects (ATE). Since each experiment is conducted on a regular basis, we obtain sequences of data for multiple experiments. For example, the CHAT\_CHAT\_SEND\_\_\_\_chat\_video\_send metric measures the total number of chats sent in the format of videos. Since we have a large repository of sequential A/B tests, in which the metric of interest is measured for both the control and treatment group for three months, the standardized $z$-scores for these tests satisfy our theoretical set up in Section \ref{setup}, i.e. coming from $F_0$ if $\text{ATE}= 0$, and $F_i$ otherwise for experiment $i$. 

\begin{tabular}[t]{cc}
    \begin{minipage}[t]{.29\linewidth}
    \centering
    \captionof*{table}{Setting 1}
        \begin{tabular}[t]{l|l}
\hline
$\pmb{\mu}$ & \pmb{Prop.} \\ \hline
0     & 0.97       \\
1.61  & 0.03      
\end{tabular}
    \end{minipage} \;
    \begin{minipage}[t]{.29\linewidth}
    \centering
    \captionof*{table}{Setting 2}
\begin{tabular}{ll}
\hline
\multicolumn{1}{l|}{\pmb{$\mu$}} & \pmb{Prop.} \\ \hline
\multicolumn{1}{l|}{0}     & 0.98       \\ 
\multicolumn{1}{l|}{1.32}  & 0.01       \\
\multicolumn{1}{l|}{1.6}   & 0.01      
\end{tabular}
    \end{minipage} \;
    \begin{minipage}[t]{.229\linewidth}
    \centering
    \captionof*{table}{Setting 3}
    \begin{tabular}{ll}
\hline
\multicolumn{1}{l|}{\pmb{$\mu$}} & \pmb{Prop.} \\ \hline
\multicolumn{1}{l|}{0}     & 0.96       \\
\multicolumn{1}{l|}{3.35}  & 0.01       \\
\multicolumn{1}{l|}{3.36}  & 0.01       \\
\multicolumn{1}{l|}{17.81} & 0.01       \\
\multicolumn{1}{l|}{19.88} & 0.01      
\end{tabular}
\end{minipage}
\end{tabular}

By pooling these experiments together, we can learn from historical data and gather useful information on the underlying metric of interest, i.e. the non-null effect size, and the overall pooled proportion of nonzero effects. The estimates of effect sizes with corresponding proportions for each of the three metrics are summarized in the above tables. For data generation, we fix the effect size at each coordinate by drawing from a multinomial distribution specified in the table. Then at each time stage, we sample independently from $\mathcal{N}(\mu, 1)$, where $\mu$ is the specified effect size. 

From results presented in Figure \ref{snap-setting1}-\ref{snap-setting3}, it can be similarly observed that the always valid p-value based method is much more conservative in term of $\mbox{FDR}$ control and has generally lower power than MSET procedure for all three real data settings. While it has to be acknowledged that the validity of data-driven version of AMSET/MSET is highly dependent on the estimation accuracy of population level parameters, our results show that the data-driven method performs reliably in most cases and yield power gain compared to other methods.

\section{Concluding Remarks}
The big data era brings tremendous opportunities and also challenges in A/B testing. 
While enjoying the potential of rich insights generated at unprecedented scale, it is crucial to be aware of the harm of ad-hoc interpretation such as ``hunt" for significance. Thus the key is the build a framework that can efficiently uncovers good insights from data, while remaining immune to the common practice of peeking. 

This paper proposes an empirical Bayes procedure AMSET for multiple testing under sequential data settings. AMSET is shown to have significant power gains as compared to existing method by leveraging rich amount of historical data available, and at the same time allows continuous monitoring without sacrificing mFDR control. 

There are multiple possible future directions based on this work. First there is still possible gap between the compound thresholding rule and the Bayes optimal threshold. Besides, more exploration in the consistency of data-driven estimation methods can be done to quantify the gap between oracle and data-driven procedures. 

\clearpage
\singlespacing
\bibliographystyle{chicago}
\bibliography{AMSET.bbl}

\clearpage
\newpage
\appendix
\section{Proof of results}
\subsection{Proof of Theorem \ref{classification}}
\begin{proof}
By the theorem set up, the posterior distribution of $\theta_i$ given $\mathbf{x}^{1:T} = \{x_{it}\}_{i \in [m], t \leq T}$ is  
\begin{equation*}
    P_{\theta_i \vert \mathbf{x}^{1:T}}(\theta_i \vert \mathbf{x}^{1:T}) = 
    \frac{(1-p)(1-\theta_i)\prod_{t=1}^T f_0(x_{it}) + p\theta_i\prod_{t=1}^T f_1(x_{it})}{(1-p)\prod_{t=1}^T f_0(x_{it})+ p\prod_{t=1}^T f_1(x_{it})}. 
\end{equation*}
Therefore the posterior risk is 
\begin{align*}
    &\mathbb{E}_{\theta\vert \mathbf{x}^{1:T}} L_\lambda(\theta, \delta^T)\\ &= \frac{1}{m}\sum_{i=1}^m  \frac{\lambda(1-p)\delta_i^T \prod_{t=1}^T f_0(x_{it}) + p(1-\delta_i^T) \prod_{t=1}^T f_1(x_{it})}{(1-p)\prod_{t=1}^T f_0(x_{it})+ p\prod_{t=1}^T f_1(x_{it})} - \alpha\lambda \delta_i^T \\
    &= \frac{1}{m}\sum_{i=1}^m \frac{p\prod_{t=1}^T f_1(x_{it})}{(1-p)\prod_{t=1}^T f_0(x_{it})+ p\prod_{t=1}^T f_1(x_{it})} \\
    & \quad + \frac{1}{m}\sum_{i=1}^m \frac{(1-\alpha)\lambda(1-p)\prod_{t=1}^T f_0(x_{it}) - (1+\alpha \lambda)p \prod_{t=1}^T f_1(x_{it})}{(1-p)\prod_{t=1}^T f_0(x_{it})+ p\prod_{t=1}^T f_1(x_{it})} \delta_i^T.
\end{align*}
Notice that the posterior risk is minimized by the decision rule 
\begin{align*}
    \delta_i^B(\lambda, T) &= 1\Big\{(1-\alpha)\lambda(1-p)\prod_{t=1}^T f_0(x_{it}) - (1+\alpha \lambda)p \prod_{t=1}^T f_1(x_{it})  \leq 0 \Big\}\\
    &=1\Big\{\frac{(1-p)\prod_{t=1}^Tf_0(x_{it})}{p\prod_{t=1}^Tf_1(x_{it})} \leq \frac{1+\alpha\lambda}{(1-\alpha)\lambda}\Big\}\\
    &=1\Big\{\frac{(1-p)\prod_{t=1}^Tf_0(x_{it})}{(1-p)\prod_{t=1}^T f_0(x_{it})+p\prod_{t=1}^Tf_1(x_{it})} \leq \frac{1+\alpha\lambda}{1+\lambda}\Big\}
\end{align*}
This concludes the proof.
\end{proof}

\subsection{Proof of Theorem \ref{thm-optimality}}
\begin{proof}
Notice that the constraint of mFDR control can also be written as 
\begin{equation*}
    \mathbb{E}\left[\sum_{i=1}^m(1-\theta_i)\delta_i\right]- \alpha \mathbb{E}\left[\sum_{i=1}^m \delta_i\right] \leq 0.
\end{equation*}
Therefore the Lagrangian dual of the optimization problem in \eqref{optimization} is exactly the classification risk for some $\lambda > 0$ 
\begin{equation}
\label{lagrangian}
    \mathcal{L}(\delta, \lambda)=\frac{1}{m}\mathbb{E}\left[\sum_{i=1}^m \theta_i(1-\delta_i) + \lambda(1-\theta_i) \delta_i - \lambda\alpha \delta_i\right].
\end{equation}
From Theorem \ref{classification} we know that $\delta^B(\lambda, T)$ is also the optimal solution for the Lagrangian dual problem in \eqref{lagrangian} for given $\lambda$. Since $\lambda >0$, we also have $\mathcal{L}(\delta^B(\lambda, T), \lambda) \leq s^*$,
where $s^*$ is the optimal value of the original constraint optimization problem \eqref{optimization}. It remains to show strong duality holds.

We claim that it suffices to show there exists $\lambda^* >0$ such that $\text{mFDR}(\delta^B(\lambda^*, T)) = \alpha$. Assume otherwise, there is some other decision $\tilde{\delta}$ such that $\text{mFDR}(\tilde{\delta}) \leq \alpha$ and $\sum_{i=1}^m\mathbb{E} \left[\theta_i (1-\tilde{\delta}_i)\right] < \sum_{i=1}^m\mathbb{E} \left[\theta_i (1-\delta^B_i(\lambda^*, T))\right]$. Then 
\begin{align*}
    &\mathcal{L}(\tilde{\delta},\lambda^*) \\
    & = \frac{1}{m}\mathbb{E}\left[\sum_{i=1}^m \theta_i(1-\tilde{\delta}_i) + \lambda^*(1-\theta_i)\tilde{\delta}_i-\lambda^*\alpha \tilde{\delta}_i\right]\\
    &< \frac{1}{m}\mathbb{E}\left[\sum_{i=1}^m \theta_i(1-\delta_i^B(\lambda^*, T)) + \lambda^*(1-\theta_i)\tilde{\delta}_i-\lambda^*\alpha \tilde{\delta}_i\right]\\
    &\leq \frac{1}{m}\mathbb{E}\left[\sum_{i=1}^m \theta_i(1-\delta_i^B(\lambda^*, T)) + \lambda^*(1-\theta_i)\delta_i^B(\lambda^*, T) - \lambda^*\alpha \delta_i^B(\lambda^*, T)\right]\\
    &= \mathcal{L}(\delta^B(\lambda^*, T), \lambda^*).
\end{align*}
This is a contradiction since $\delta^B(\lambda^*, T)$ is the minimizer of the Lagrangian with $\lambda^*$. 

We then proceed to show that the constraint equality can be satisfied. First notice that we can write 
\begin{equation*}
    \text{mFDR}(\delta^B(\lambda, T)) = \frac{(1-p)G_0^{iT}\left(\frac{1+\alpha\lambda}{1+\lambda}\right)}{(1-p)G_0^{iT}\left(\frac{1+\alpha\lambda}{1+\lambda}\right) + pG_1^{iT}\left(\frac{1+\alpha\lambda}{1+\lambda}\right)}.
\end{equation*}
As $\lambda \to 0$, $G_0^{iT}\left(\frac{1+\alpha\lambda}{1+\lambda}\right), G_1^{iT}\left(\frac{1+\alpha\lambda}{1+\lambda}\right) \to 1$. Therefore $\text{mFDR}(\delta^B(\lambda, T)) \to 1-p > \alpha$. As $\lambda \to +\infty$, $G_0^{iT}\left(\frac{1+\alpha\lambda}{1+\lambda}\right) \to G_0^{iT}(\alpha), G_1^{iT}\left(\frac{1+\alpha\lambda}{1+\lambda}\right) \to G_1^{iT}(\alpha)$. Therefore $\text{mFDR}(\delta^B(\lambda, T)) \to \frac{(1-p)G_0^{iT}(\alpha)}{(1-p) G_0^{iT}(\alpha)+ pG_1^{iT}(\alpha)} < \alpha$ by assumption. By continuity of $\text{mFDR}(\delta^B(\lambda, T))$, there exists $\lambda^*>0$ such that $\text{mFDR}(\delta^B(\lambda^*, T)) = \alpha$. 
\end{proof}

\subsection{Proof of Proposition \ref{prop-optimality-condition}}
\begin{proof}
Notice that for the simple truncation decision rule $\delta^T = (\delta_1^T,\ldots, \delta_m^T)$, where $\delta_i^T = 1\{T_{OR}^{iT} \leq \alpha\}$, 
\begin{equation*}
    \text{mFDR}(\delta^T) = \frac{\mathbb{E}\left[\sum_{i=1}^m (1-\theta_i)\delta^T_i\right]}{\mathbb{E}\left[\left(\sum_{i=1}^m \delta_i^T \right) \vee 1\right]} = \frac{(1-p)G_0^{iT}(\alpha)}{(1-p)G_0^{iT}(\alpha) + p G_1^{iT}(\alpha)}.
\end{equation*}
For the numerator, 
\begin{align*}
    \mathbb{E}\left[\sum_{i=1}^m (1-\theta_i)\delta^T_i\right]
    &= \mathbb{E}\left[\sum_{i=1}^m \mathbb{E}\left[(1-\theta_i)\delta_i^T \, \Big\vert \, x_{i,1:T}\right]\right]\\
    &= \mathbb{E}\left[\sum_{i=1}^m \delta_i^T \mathbb{E}\left[1-\theta_i \, \Big\vert \, x_{i,1:T}\right]\right]\\
    &= \mathbb{E}\left[\sum_{i=1}^m \delta_i^T T_{OR}^{iT}\right] \leq \alpha \mathbb{E}\left[\sum_{i=1}^m \delta_i^T\right],
\end{align*}
where the last inequality is due to definition of the truncation rule. Therefore $\text{mFDR}(\delta^T) \leq \alpha$. 
\end{proof}

\subsection{Optimality of adaptive procedure}
\label{appendix-adaptive-optimality}
\begin{lemma}
\label{lem-optimality}
Under the same data generating process and conditions in Theorem \ref{thm-optimality}. For any time stage $T$, conditioning on the active set $\mathcal{A}_T \in \mathcal{F}_{T-1}$, there exists a $\lambda^* > 0$ such that $\delta^B(\lambda^*, T)$ is the stage-wise optimal test in the sense that it solves the constraint optimization problem 
\begin{align}
\begin{split}
\label{optimization-adaptive}
    \min_{\{\delta_i\}_{i \in \mathcal{A}_T}} & \quad   \mathbb{E}\left[\sum_{i\in \mathcal{A}_T}\theta_i (1-\delta_i)\right]\\
    \textrm{s.t.} & \quad \text{mFDR}(\delta)= \frac{\mathbb{E}\left[\sum_{i\in \mathcal{A}_T}(1-\theta_i)\delta_i\right]}{\mathbb{E}\left[\left(\sum_{i\in \mathcal{A}_T} \delta_i \right)\vee 1\right]} \leq \alpha.
\end{split}
\end{align}
Moreover, if we let $\delta^T = (\delta_1^T,\ldots, \delta_m^T)$ where $\delta_i^T = 1$ for $i \notin \mathcal{A}_T$ and $\delta_{\mathcal{A}_T^c}^T = \delta^B(\lambda^*, T)$, then we also have global mFDR control at stage $T$, i.e. $mFDR(\delta^T) \leq \alpha$.
\end{lemma}
\begin{proof}
Given an active set $\mathcal{A}_T$, the coordinates for test at stage $T$ are fixed and are sampled at all past time stages. Therefore by Theorem \ref{thm-optimality} with $m = |\mathcal{A}_T|$, we have directly the first part of the constraint optimization problem. 

For global mFDR control, notice that 
\begin{align*}
 \text{mFDR}(\delta^T) = \frac{\mathbb{E}\left[\sum_{i=1}^m (1-\theta_i)\delta^T_i\right]}{\mathbb{E}\left[\left(\sum_{i=1}^m \delta_i \right) \vee 1\right]} = \frac{\mathbb{E}\left[\sum_{t = 1}^T\sum_{i \in \mathcal{A}_t} (1-\theta_i)\delta_i^T\right]}{\mathbb{E}\left[\sum_{t=1}^T \left(\sum_{i\in \mathcal{A}_t} \delta_i\right) \vee 1\right]}.
\end{align*}
Given stagewise mFDR control
$\mathbb{E}\left[\sum_{i \in \mathcal{A}_t} (1-\theta_i)\delta_i^t\right] \leq \alpha \mathbb{E}\left[\sum_{i \in \mathcal{A}_t} \delta_i^t\right]$, we have global mFDR control. Since the choice of each $\lambda^*$ in all past time stages controls stagewise mFDR, we have global mFDR control. 
\end{proof}

\subsection{Proof of Theorem \ref{thm-oracle}}

\begin{proof}
To show validity of AMSET for mFDR control at any stopping time $\tau$, we borrow a similar idea from \cite{wang2017sparse}. Let $\mathbf{x}^{1:T} = \{x_{it}\}_{i \in \mathcal{A}_t, t \leq T}$ be all observations in past stages before time $T$ and $\mathcal{F}^T = \sigma(\mathbf{x}^{1:T})$. Recall the definition of mFDR in the sequential context for stopping time $\tau$. 
\begin{align*}
    \text{mFDR}(\tau) &= \frac{\mathbb{E}\left[\sum_{t=1}^\tau \sum_{i \in S_t} (1-\theta_i) \delta_i^t\right]}{\mathbb{E}\left[\left(\sum_{t=1}^\tau \sum_{i \in S_t} \delta_i^t\right) \vee 1\right]}.
\end{align*}
Note that the numerator can be written as, 
\begin{align*}
    \mathbb{E}\left[\sum_{t=1}^\tau \sum_{i \in S_t} (1-\theta_i) \delta_i^t\right] & = \mathbb{E}\left[\sum_{t=1}^\infty \sum_{i \in S_t} (1-\theta_i) \delta_i^t 1\{t \leq \tau\}\right]\\ 
    & = \sum_{t=1}^\infty \mathbb{E}\left[\sum_{i \in S_t} (1-\theta_i) \delta_i^t 1\{t \leq \tau\}\right]\\
    & \stackrel{(1)}{=} \sum_{t=1}^\infty  \mathbb{E}\left[\mathbb{E}\left[\sum_{i \in S_t} (1-\theta_i) \delta_i^t 1\{t \leq \tau\} \, \Big\vert \, \mathcal{F}^t\right]\right]\\
    & \stackrel{(2)}{=} \sum_{t=1}^\infty  \mathbb{E}\left[1\{t \leq \tau\}\mathbb{E} \left[\sum_{i \in S_t} (1-\theta_i) \delta_i^t \, \Big\vert \, \mathcal{F}^t\right]\right]\\
    & \stackrel{(3)}{=} \sum_{t=1}^\infty  \mathbb{E}\left[1\{t \leq \tau\}\sum_{i \in S_t} \delta_i^t\mathbb{E} \left[(1-\theta_i)  \, \Big\vert \, \mathcal{F}^t\right]\right]\\
    & \stackrel{(4)}{=} \sum_{t=1}^\infty  \mathbb{E}\left[1\{t \leq \tau\}\sum_{i \in S_t} \delta_i^t T_{OR}^{it} \right]\\
    &\stackrel{(5)}{\leq}  \alpha \sum_{t=1}^\infty  \mathbb{E}\left[1\{t \leq \tau\}\sum_{i \in S_t} \delta_i^t \right] = \alpha \mathbb{E}\left[\sum_{t=1}^\tau \sum_{i \in S_t} \delta_i^t\right],
\end{align*}
where (1) is by tower property, (2) and (3) are by the fact that both $1\{t\leq \tau\}$ and $\delta_i^t$ are $\mathcal{F}^t-$measurable, (4) is by definition of lfdr $T_{OR}^{it}$ in Eq. \eqref{eq-Tor}, and (5) is by compound thresholding property in Algorithm \ref{algo-oracle}. Combining the upper bound with the denominator of mFDR, we have that $\text{mFDR} \leq \alpha$ for any stopping time $\tau$. This result ensures validity up to continuous monitoring and therefore immunity to peeking.

Stagewise mFDR control is immediate by the same proof as above without summation over time stages. For stagewise exact FDR control, we have that for any time stage $t$
\begin{align*}
    \text{FDR}(\delta^t | S_t) &= \mathbb{E}\left[\frac{\sum_{i \in S_t} (1-\theta_i) \delta_i^t}{\left(\sum_{i \in S_t} \delta_i^t\right) \vee 1}\right] \\
    &= \mathbb{E}\left[\sum_{i \in S_t} \mathbb{E}\left[\frac{(1-\theta_i) \delta_i}{\left(\sum_{i \in S_t} \delta_i^t\right) \vee 1} \,\Big\vert\, \mathcal{F}^t\right]\right] \\
    &= \mathbb{E}\left[ \frac{\sum_{i \in S_t}\delta_i T_{OR}^{it}}{\left(\sum_{i \in S_t} \delta_i^t\right) \vee 1}\right] \leq \alpha \mathbb{E}\left[ \frac{\sum_{i \in S_t}\delta_i}{\left(\sum_{i \in S_t} \delta_i^t\right) \vee 1}\right] \leq \alpha.
\end{align*}
\end{proof}

\section{Additional procedures}
We provide details of some additional procedures in the paper, including non-adaptive version of AMSET called multistage empirical Bayes test (MSET), and Gaussian mixture deconvolution for estimation of alternative density. 

\subsection{MSET}
The non-adaptive MSET procedure is described in Algorithm \ref{algo-oracle-mset}.
\RestyleAlgo{ruled}
\begin{algorithm}[!h]
\caption{MSET}
\label{algo-mset}
Specify any stopping criterion\;
\For{time stages $t= 1, 2,\ldots$}{
(\textbf{Data Sampling}) Observe data $\{x_{it}\}_{i \in [m]}$\;
(\textbf{Ranking}) For all $i$, compute $T_{OR}^{it}$ in Eq \eqref{eq-Tor} based on $\{x_{is}\}_{s \leq t}$ and sort them in ascending order 
$$T_{OR}^{(1),t}, T_{OR}^{(2),t}, \ldots, T_{OR}^{(m),t}\;$$
(\textbf{Compound Thresholding}) Define threshold 
\begin{equation*}
r_t = \max\Big\{r: \frac{1}{r} \sum_{i= 1}^{r} T_{OR}^{(i),t} \leq \alpha\Big\}
\end{equation*}
 Let $S_{t} = \{i \in [m]: T_{OR}^{it} \leq T_{OR}^{(r_t), t}\}$ \;
(\textbf{Decision}) Let $\delta_i^t = 1$ for all $i \in S_t$ and $0$ otherwise \;
Break if stopping criterion is satisfied.
}
\end{algorithm}

\subsection{Gaussian mixture deconvolution}
\label{appen-deconvolve}
In Section \ref{sec-dd}, we described a data-driven method for AMSET, where we replace the oracle lfdr test statistic in Eq. \eqref{eq-Tor} with estimated population level parameters to form the data-driven test statistic in Eq. \eqref{eq-Tdd}. In particular for the estimation of non-null density, we will use Kiefer-Wolfowitz maximum likelihood estimation for Gaussian mixtures with software ``REBayes" develped by \cite{koenker2017rebayes} and computation supported by \cite{koenker2014convex}.

Recall that Gaussian mixture deconvolution gives us marginal density estimation as a mixture of Gaussians $\hat{G} = \sum_{i=1}^k \hat{p}_i \hat{G}_i$, where $\sum_{i=1}^k \hat{p}_i = 1$ and $\hat{G}_i = \mathcal{N}(\hat{\mu}_i, 1)$ for some $k$. There are various methods for recovering an estimate of $\hat{f_1}$ from deconvolution results. We list three possible methods below. 

\begin{itemize}
    \item Impose a hard threshold $c$ such that we only include $\hat{G}_i$'s such that $|\hat{\mu}_i| \geq c$ with their weighted corresponding proportions.
    \item Use the classification boundary in homescedastic Gaussian mixture model in \cite{donoho2004higher} as a sanity check for estimation of densities. That is, given an estimate of the proportion of non-nulls, we can find the smallest detectable effect from the classification boundary and consider only detectable effects as potential candidates for alternative density estimation.
    \item Estimate threshold to recover alternative density in a completely data-driven way. 
 That is, we first sort components of the Gaussian mixture estimate in ascending order of $|\hat{\mu}_i|$'s, resulting in $\hat{G}_{(1)}, \ldots, \hat{G}_{(k)}$ with their corresponding proportions $\hat{p}_{(1)}, \ldots, \hat{p}_{(k)}$. We can place a threshold $k_a = \max\{k:\sum_{i=1}^k \hat{p}_{(i)} \leq 1-\hat{p}\}$, where $\hat{p}$ is the proportion of non-nulls estimated by step (1) of data-driven implementation in Section \ref{sec-dd}. The alternative density can then be estimated by $\hat{f}_1 = \sum_{i = k_a}^k \hat{p}_{(i)}\hat{G}_{(i)}/\sum_{i = k_a}^k \hat{p}_{(i)}$. 
\end{itemize}

\section{Reproducibility}
We provide R implementations of all models described and the reproduction of all simulation study results at https://github.com/AMSET2022/amset. This supplementary section
details the steps to reproduce our results. 

\subsection{Code Organization}
All R codes for generating simulation results are included in the folder ``R". The csv files and plots for our simulations in fixed horizon comparison and stagewise comparison can be found in the ``data" folder. Since the package REBayes that provides Kiefer-Wolfowitz maximum likelihood estimation for mixture models deconvolution requires the installation of the optimization package Rmosek. We also provided links to user guide and resources for installing Rmosek. 

\subsection{Simulation code}
The details for simulation code are as follows. 
\begin{itemize}
    \item \verb|epsest.func.R| contains auxiliary function to estimate the proportion of non-null locations.
    \item \verb|MSET.R| implements fixed horizon comparison between MSET procedure and Optimizely's method using always valid p-values.
    \item \verb|MSET_stagewise.R| implements stagewise comparison between MSET procedure and Optimizely's method using always valid p-values.
    \item \verb|AMSET.R| implements fixed horizon comparison between AMSET procedure and Optimizely's method using always valid p-values. 
\end{itemize}

\end{document}